\documentclass{article}
\usepackage{graphicx} % Required for inserting images

\usepackage{float}

\usepackage{amsfonts}
\usepackage{amsmath}
\usepackage{amssymb}
\usepackage{amsthm}
\usepackage{braket}
\usepackage{gensymb}
\usepackage{tikz-cd}
\usepackage{stmaryrd}
\usepackage[backend=bibtex,style=numeric,sorting=nty,backref]{biblatex}
\usepackage[colorlinks=true, linkcolor=black, citecolor=black, urlcolor=blue]{hyperref}

\theoremstyle{plain}
\newtheorem{thm}{Theorem}
\newtheorem{lem}{Lemma}
\newtheorem{cor}{Corollary}
\newtheorem*{thmA}{Theorem A}
\newtheorem*{condB}{Condition B}
\newtheorem*{mconj}{Main Conjecture}

\newcommand{\ie}{i.\,e.\ }
\newcommand{\lmod}{{\operatorname{-mod}}}
\newcommand{\tr}{\mathrm{Tr}}
\newcommand{\Ext}{\mathrm{Ext}}
\newcommand{\Id}{\mathrm{Id}}
\newcommand{\Image}{\mathrm{Im}}
\newcommand{\End}{\mathrm{End}}
\newcommand{\Hom}{\mathrm{Hom}}

\title{Statistical Mechanics and Categorical Entropy}
\author{
  Haiqi Wu\\
  Computer Laboratory\\
  University of Cambridge\\
  Cambridge, UK\\
\texttt{\detokenize{hw625@cam.ac.uk}}
  \and
  Kai Xu\\
  Department of Physics\\
  Harvard University \\
  Massachusetts, United States\\
\texttt{\detokenize{k_xu@g.harvard.edu}}
  }
\date{\today}%{October 2024}
\begin{document}

\maketitle
\begin{abstract}
    \noindent This paper investigates the relationship between categorical entropy and von Neumann entropy of quantum lattices. We begin by studying the von Neumann entropy, proving that the average von Neumann entropy per site converges to the logarithm of an algebraic integer in the low-temperature and thermodynamic limits. Next, we turn to categorical entropy. Given an endofunctor of a saturated $A_{\infty}$-category, we construct a corresponding lattice model, through which the categorical entropy can be understood in terms of the information encoded in the model. Finally, by introducing a gauged lattice framework, we unify these two notions of entropy. This unification leads naturally to a sufficient condition for a conjectural algebraicity property of categorical entropy, suggesting a deeper structural connection between $A_{\infty}$-categories and statistical mechanics.
\end{abstract}
\newpage
\tableofcontents
\newpage
\section{Introduction}
\noindent
This paper aims at understanding the algebraicity conjecture of categorical entropy \cite{1307.8418} of endofunctors on saturated $A_{\infty}$-categories \cite{0606241v3}. 
\\
\newline
\noindent
The study of \( A_\infty \)-categories \cite{0510508, 0606241v3} 
and categorical entropy \cite{1307.8418}
has broad applications across mathematics and physics. 
Categorical entropy, as introduced in the context of dynamical systems and category theory \cite{1307.8418}, measures the complexity of transformations in a category. 
Loosely speaking, to an endofunctor $F$ on a triangulated category, one associates a function $h_t(F)$ that describes the growth of complexity as the number of times $F$ is applied. 
In symplectic geometry, for example, the Fukaya category of a symplectic manifold provides an \( A_\infty \)-category \cite{HoriMS} where categorical entropy can analyze symplectic automorphisms and pseudo-Anosov maps \cite{Blanc_2015}, 
offering a categorical lens on dynamical systems traditionally understood through topological entropy \cite{adlertopologicalentropy}. 
In algebraic geometry, categorical entropy applies to autoequivalences in derived categories of coherent sheaves \cite{orlov2009derivedcategoriescoherentsheaves},
enabling the study of dynamical behavior of birational maps \cite{Cantat_2013, Blanc_2015} and stability conditions \cite{bridgeland2006stabilityconditionstriangulatedcategories}. By embedding the traditional notion of entropy within the homological and higher categorical structure, categorical entropy allows for a nuanced understanding of complexity in a categorical setting, with implications for understanding stability, transformations, and the homotopy properties of categories across various mathematical and physical disciplines. 
\newline
\noindent
\\
It has been conjectured \cite{1307.8418} that in a saturated $A_{\infty}$ category, $\exp(h_0(F))$, the exponential of the entropy $h_t(F)$ at the value $t=0$, is an algebraic integer. It was asked what the natural sufficient conditions are for this to hold. \\
\newline
\noindent
On the other hand, consider a one-dimensional lattice $L$ of $N$ sites, each associated with a vector space $\mathbb{C}^2\cong \mathrm{span}_{\mathbb{C}}(\ket{1},\ket{0})$. Suppose the evolution of this lattice system is described by the \textbf{Fibonacci Hamiltonian} acting locally on each pair of adjacent sites of the lattice as:
\[
  \mathcal{H}\ket{xy} =
  \begin{cases}
    \ket{11} & \text{if } x = y = 1, \\
    0 & \text{otherwise}.
  \end{cases}
\]
\noindent
In \cite{wang2024latticemodelfibonaccidegree}, it is shown that the average degree of degeneracy per site of this system, $\frac{1}{N}\dim \ker(\mathcal{H})$, is an algebraic integer in the limit $N\rightarrow\infty$. In this paper, we will generalize this result by means of the von Neumann entropy \cite{vonNeumannQSM} of a quantum lattice system. The result we obtain provides an intriguing relationship between statistical mechanics and number theory. 
\\
\newline
\noindent
We try to unify these two above-mentioned concepts. We will use tools developed in quantum lattice models to formulate a sufficient condition for the algebraicity conjecture in the $A_{\infty}$-category context. As we shall see, this condition turns out to be a common generalization of the algebraicity statement about both the lattice (von Neumann) entropy and the categorical entropy.
\\
\newline
\noindent
In Section 2, we discuss the von Neumann entropy of a quantum lattice. We prove that, when the number of sites in a lattice goes to $\infty$, the exponential of the average entropy per site is an algebraic integer. We summarize the result at the end of this section as Theorem A.\\%, generalizing the example we mentioned above.\\
\newline
\noindent
In Section 3, we talk about the categorical entropy in the saturated $A_{\infty}$-category setting. We will first extract a lattice model from the information of an endofunctor, and then specify a sufficient condition for the algebraicity conjecture of $\exp(h_0(F))$ naturally induced from the lattice model. We call this Condition B.
\\
\newline
\noindent
We finish Section 3 with gauged lattice models. We propose our main conjecture about the von Neumann entropy of a gauged lattice. As a result, the algebraicity statements about the quantum lattice and about the saturated $A_{\infty}$-category are both special cases of the main conjecture:
% We will see that the algebraicity statements about the quantum lattice and about saturated $A_{\infty}$-algebra are both special cases of the situation when one considers the von Neumann entropy of a gauged lattice. In particular, we propose a statement \textit{Claim X}, which is a common generalization of both \textit{Claim A} and \textit{Claim B}.
% \begin{centering}
\[
\begin{tikzcd}[ column sep=small, scale=0.9]
& \text{Main Conjecture}\arrow[Rightarrow]{dl}{}\arrow[Rightarrow]{dr}{}   \\
\text{Theorem A}  &  & \text{Condition B}\text{ (conjectured)}\arrow[Rightarrow]{d}{}\\
 & & \shortstack{\textit{algebraicity cojecture} \\ \textit{of $\exp(h_0(F))$}}\\
\end{tikzcd}
\]
% \end{centering}
% \[
% \begin{tikzcd}[row sep=small, column sep=small, scale=0.9]
% & \textit{Claim X (conjectured)}
%   \arrow[Rightarrow, shorten <=4pt, shorten >=4pt]{dl}[swap]{\text{implies}}
%   \arrow[Rightarrow, shorten <=4pt, shorten >=4pt]{dr}{\text{implies}} \\
% \textit{Claim A (proved)} &  & \textit{Claim B (conjectured)} \arrow[Rightarrow]{d}{\text{implies}}\\
%  & & \textit{algebraicity of $\exp(h_0(F))$}
% \end{tikzcd}
% \]
\\
\noindent
\textbf{Acknowledgements.} We express our gratitude to Nils Lauermann, Ricardo Ali, Thibaut Benjamin, Ioannis Markakis for many insightful discussions, and to Prof. Jamie Vicary for the invaluable feedback. Special thanks to Prof. Dame Caroline Humphrey for her help during the author's most difficult time. 
\section{Lattice model and the von Neumann entropy}
\noindent
In statistical mechanics \cite{Baxter:1982zz}, entropy serves as a key concept in understanding the probabilistic nature of thermodynamic systems. Formally introduced by Boltzmann \cite{Boltzmann_2012}, entropy $S$ is defined in terms of the number of possible microscopic configurations (or microstates) that correspond to a given macroscopic state (or macrostate) of a system. Mathematically, \textbf{the Boltzmann entropy} is expressed as 
$$S=k_B\log \Omega= -k_B \sum_i P_i \log P_i $$ 
where $\Omega$ represents the count of accessible microstates in case they are equally distributed; $P_i$ represents their probability in general and $k_B$ is Boltzmann's constant, which may be set to $1$ in natural units. Entropy quantifies the level of uncertainty or randomness associated with the exact arrangement of particles, given only the macroscopic variables such as temperature, volume, and pressure. As the system evolves, it tends to move towards states of higher entropy, reflecting an increase in disorder and aligning with the second law of thermodynamics. This movement toward equilibrium is fundamentally linked to the probabilistic nature of particle interactions, providing a bridge between microscopic dynamics and observable macroscopic phenomena. Entropy thus becomes not only a measure of disorder but also a driving force behind irreversible processes in physical systems, playing a pivotal role in understanding the arrow of time and the behavior of matter in various states.
\noindent
\newline
\\
In quantum statistical mechanics \cite{vonNeumannQSM,Bratteli:1979tw}, the concept of entropy extends to accommodate the probabilistic nature of quantum states and their behavior under the laws of quantum mechanics. Quantum mechanical entropy, often referred to as von Neumann entropy, is a measure of the uncertainty or information content associated with the states of a quantum system. For a quantum system represented by a nonnegative definite density matrix $\rho$ with trace $1$, \textbf{the von Neumann entropy} $S$ is defined as 

$$S=-\tr(\rho\log \rho)$$
where $\tr$ denotes the trace operator. This definition generalizes the classical notion of entropy, incorporating the fact that quantum states can exist in superpositions, and the system may not be in a definite state until a measurement is performed. 
%Unlike classical entropy, quantum entropy captures the inherent probabilistic nature of quantum states, including entanglement in composite systems, where the entropy can reflect the degree of entanglement between subsystems. This makes quantum entropy an essential concept for understanding quantum information theory, quantum thermodynamics, and the behavior of quantum systems at different temperatures and energy levels.
\newline
\\
\noindent
Based on the concepts of entropy in both classical and quantum mechanics, we can understand the Boltzmann ensemble, which underpins the statistical description of systems in thermal equilibrium, and emerges naturally when we seek to maximize entropy for a system with fixed constraints, such as energy.
\newline
\\
\noindent
In a classical setting, if we consider a system in equilibrium with a heat reservoir, it is subject to a fixed average energy constraint. To determine the most likely distribution of particles across states, we maximize the entropy \( S = -k_B \sum_i P_i \log P_i \), where \( P_i \) is the probability that the system is in the \(i\) -th state with energy \( E_i \), subject to the constraint \( \sum_i P_i E_i = U \), where \( U \) is the average energy. This maximization leads to the Boltzmann distribution:
\[
P_i = \frac{e^{-E_i / k_B T}}{Z}
\]
where \( T \) is the temperature (which appears as the natural Lagrange multiplier), \( k_B \) is Boltzmann’s constant, and \( Z = \sum_i e^{-E_i / k_B T} \) is the partition function. This distribution describes the most probable state of the system in thermal equilibrium, with higher energy states being less probable as they contribute less to the overall entropy.
\\
\newline
\noindent
In a quantum setting, a similar approach applies to systems in equilibrium. Here, the density matrix \( \rho \) captures the distribution of the quantum states, and the entropy is given by the von Neumann formula \( S = -k_B \operatorname{Tr}(\rho \log \rho) \). To determine the density matrix that maximizes entropy while keeping the average energy fixed, we maximize \( S \) under the constraint \( \operatorname{Tr}(\rho \mathcal{H}) = U \), where \( \mathcal{H} \) is the Hamiltonian of the system and \( U \) is the average energy. This procedure yields the quantum analog of the Boltzmann distribution:
\[
\rho = \frac{e^{-\mathcal{H} / k_B T}}{Z}
\]
where \( Z = \operatorname{Tr}(e^{-\mathcal{H} / k_B T}) \) is the quantum partition function. This density matrix provides the equilibrium distribution over quantum states, reflecting the tendency of the system to occupy states that maximize entropy within the energy constraint.
\\
\newline
\noindent
Thus, in both classical and quantum settings, the Boltzmann ensemble represents the state of maximum entropy for a given energy. By maximizing entropy, the system naturally adopts a distribution in which higher-energy states are exponentially suppressed compared to lower-energy ones. This maximization principle not only defines the equilibrium state but also underscores the intrinsic link between entropy and the probabilistic nature of statistical mechanics. The Boltzmann ensemble is therefore a natural outcome of the entropy maximization process, embodying the statistical distribution that best represents a system in thermal equilibrium. In doing so, it provides a powerful framework for calculating thermodynamic properties, enabling the study of phase transitions, fluctuations, and the macroscopic behavior of matter as it evolves towards equilibrium.

\subsection{The algebraicity in the case of a quantum lattice model}

The notion of lattice models \cite{Baxter:1982zz} is a cornerstone for classical statistical mechanics, offering a simplified framework for analyzing the collective behavior of interacting particles arranged in a regular, grid-like structure. In these models, each lattice point, or site, represents a particle or a small region, and the states of these sites are governed by a set of interactions and rules. The Ising model \cite{Ising:1925em}, one of the most well-known lattice models, serves as a classic example for studying ferromagnetism. In the Ising model, each site on the lattice is assigned a spin, either up or down, and neighboring spins interact to minimize or maximize alignment, depending on whether the interactions are ferromagnetic or antiferromagnetic. By employing techniques such as the partition function, researchers can calculate thermodynamic quantities like entropy, free energy, and magnetization to understand phase transitions \cite{Landau:1980mil}. 
%such as those between ordered and disordered states. 
%Lattice models thus provide valuable insights into critical phenomena and are used to study a wide range of systems beyond magnetism, including fluid behavior, percolation, and biological systems.
\\
\newline
\noindent
Quantum statistical mechanics  \cite{Bratteli:1979tw} adapts lattice models to account for quantum mechanical principles, allowing researchers to explore quantum phase transitions and the behavior of particles at extremely low temperatures, where quantum effects dominate. In quantum lattice models, particles are placed on a lattice, and their interactions are governed by quantum operators instead of classical probabilities. The Bose-Hubbard model \cite{BoseHubbard} and the Heisenberg model \cite{HeisenbergModel,Baxter:1982zz} are well-known examples of quantum lattice models. By studying quantum lattice models, researchers gain insights into the quantum nature of matter, particularly in systems where entanglement and quantum correlations play a significant role.
%The Bose-Hubbard model is a prominent example of a quantum lattice model, which describes particles that obey Bose-Einstein statistics. In this model, particles can hop between lattice sites and interact with one another if they occupy the same site, leading to phenomena such as superfluidity and the Mott insulator transition. Another well-known quantum lattice model is the Heisenberg model, which describes spin systems and is crucial for understanding quantum magnetism. These models are solved using advanced techniques like quantum Monte Carlo methods and tensor network approaches, enabling scientists to investigate quantum many-body systems and their complex behaviors under different conditions. By studying quantum lattice models, researchers gain insights into the quantum nature of matter, particularly in systems where entanglement and quantum correlations play a significant role.
\\
\newline
\noindent
The basic ingredients of a quantum lattice model consist of the following:
\begin{itemize}
    \item A $d$-dimensional lattice $L$ of size $N$ (Here by lattice we mean $\{1,2,\cdots,N\}^d$)

    \item For each lattice point $x\in L$, we have a local Hilbert space $V_x$ such that all of them are canonically isomorphic (here we regard, for example, $V_1=V_{N+1}$ when $d=1$, and similarly for higher dimensions). 

    Given a subset $U\subseteq L$, we define its space of states 

    $$V(U)=\bigotimes _{x\in U} V_x$$

    and the global space of states $V=V(L)$
    
    \item For some subset $U\subseteq L$ of size $m<N$, we have a local Hamiltonian (\ie a non-negative definite Hermitian operator) $\mathcal{H}_U$ acting on $V(U)$.
    Then $\mathcal{H}_U$ also naturally acts on $V$ by tensoring with the identity operators of other sites. We define the global Hamiltonian 

    $$\mathcal{H}=\sum _{U'\subseteq L \mathrm{ \,is\, a\, translation\, of\,} U} \mathcal{H}_{U'}$$

    \item (Assumption) any two local Hamiltonians $\mathcal{H}_{U'}$ and $\mathcal{H}_{U''}$ obtained as above commute, namely,
    \[
    [\mathcal{H}_{U'},\mathcal{H}_{U''}]=0,
    \] as well as 
    the operators $\mathrm{P}_{\ker \mathcal{H}_{U'}}$ and $\mathrm{P}_{\ker \mathcal{H}_{U''}}$ associated to $\mathcal{H}_{U'},\mathcal{H}_{U''}$, respectively,
    where $\mathrm{P}_{\ker \mathcal{H}_U}$ denotes the projection onto the kernel of $\mathcal{H}_U$. In the sequel, we will denote them by $\mathrm{P}_U$ with a slight abuse of notation.

    \end{itemize}
From these data, we naturally have the von Neumann entropy of this lattice model defined in the previous section. It's natural to study its behavior, but this is too complicated and it's impossible to understand its detailed behavior in full generality. 
\begin{figure}[H]
    \centering
    \includegraphics[width=7.7cm]{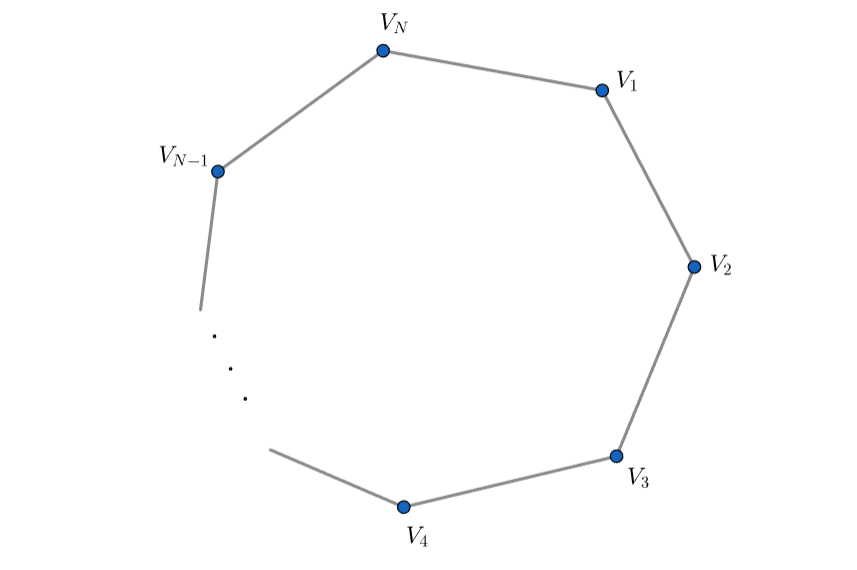}
    \caption{A one-dimensional lattice with $V_{N+1}=V_1$}
    \label{fig:graphic representation example}
\end{figure}
\noindent
Our main theorems concern the behavior of the system in the \textbf{low-temperature} and \textbf{thermodynamic limits} $ T\longrightarrow 0$ and $N\longrightarrow \infty$:

\begin{lem}
    In the low temperature limit $T\longrightarrow 0$, the von Neumann entropy is given by the logarithm of \textbf{the ground state degeneracy G}:

    $$S=\log G=\log (\dim \ker \mathcal{H})$$.
\end{lem}
\begin{proof} This follows from the identity $$\lim_{T\longrightarrow 0} e^{-\frac{\mathcal{H}}{k_BT}}=\mathrm{P}_{\ker \mathcal{H}}.$$
Let $\beta = \frac{1}{k_BT}$. Then for non-negative real number $E$,
\[   
\lim_{\beta\rightarrow\infty}e^{-\beta E} = 
     \begin{cases}
       0 &\quad\text{if }E>0\\
       1 &\quad\text{if }E=0 \\

     \end{cases}
\]
Let 
\[
  D =
  \begin{pmatrix}
    E_{1} & & \\
    & \ddots & \\
    & & E_{n}
  \end{pmatrix}
\]
be the diagonal matrix corresponds to $\mathcal{H}$, where $E_i$'s are eigenvalues of $\mathcal{H}$. Then
\[
  e^{-\beta D} =
  \begin{pmatrix}
    e^{-\beta E_1} & & \\
    & \ddots & \\
    & & e^{-\beta E_n}
  \end{pmatrix}
\]
is the diagonal matrix corresponding to $e^{-\beta \mathcal{H}}$. This diagonal matrix has entries $1$ precise where the eigenvalues of $\mathcal{H}$ are $0$.
\end{proof}
\noindent
Now that we want to study its behavior when we vary the system size $N$, we denote the entropy by $S_N$, the global Hamiltonian by $\mathcal{H}_N$, and the ground-state degeneracy by $G_N$. By the above lemma, $S_N=\log G_N$ and $G_N=\dim(\ker (\mathcal{H}_N))=\tr \mathrm{P}_{\ker \mathcal{H}_N} $.

\begin{thm}
    In the low temperature limit $T\longrightarrow 0$, the ground state degeneracy $G_N$ for $N\in\mathbb{N}^*$ satisfies a linear recurrence relation with integral coefficients

    $$\sum_{k=0}^{N_0} a_kG_{N-k}=0$$

\noindent where the order $N_0$ depends only on $\dim V$ and $m$, not on $N$.
    \end{thm}
\noindent Before proving this theorem, we introduce some additional notation.\\
First, let $L$ be a lattice of dimension $1$ with size $N$. Suppose the vector spaces associated with the local sites are $V_1$, $V_2$, $\dots$, $V_N$. Then
\[V(L)=V_1\otimes V_2\otimes\dots\otimes V_N.\]
\noindent We define the translation operator $\tau:V(L)\rightarrow V(L)$ by
\[
\tau:v_1\otimes v_2\otimes\dots\otimes v_N\mapsto v_2\otimes  v_3\otimes\dots\otimes v_N\otimes v_1
\]
Visually, this operation corresponds to a cyclic left shift of the sites. Observe that one can visit all $V_i$'s for $i\in\{1,2,\dots,N\}$ starting from $V_1$ by repeatedly applying $\tau$. Moreover, given any local operator $F$ acting on $U$ for some sublattice $U\subseteq L$ of size $m$ (we may assume, without loss of generality, that $U$ consists of the first $m$ sites), 
\begin{itemize}
    \item$\tau$ translates adjacent sites to adjacent sites, so that $F\circ\tau$ is well-defined;
    \item $\{F\circ\tau^i|i\in\mathbb{Z}\}$ covers all possible translated positions of $F$ on $L$.
\end{itemize}
We say that $\tau$ preserves adjacency with respect to $F$.
\\
In general, when $L$ is of any dimension $d$ and $U\subseteq L$, it is possible to design a translation $\tau$ on $L$, which traverses all local sites in $L$ by applying it repeatedly. Given any local operator $F$ acting on $U$, it is always possible to find a translation that traverses all sites starting from any given site in $U$ which preserves adjacency with respect to $F$ by treating each dimension separately. To illustrate, suppose for example $d=2$ and $F$ acts on a square of length $1$. Consider $\tau$ as follows: 
\[
v_{1,1}\mapsto v_{1,2}\mapsto\dots\mapsto v_{1,N}\mapsto v_{2,N}\mapsto v_{2,1}\mapsto v_{2,2}\mapsto\dots\mapsto v_{2,N-1}\mapsto v_{3,N-1}\mapsto\dots  
\]
Visually, $\tau$ traverses the \textit{diagonal} of the torus and eventually returns to $v_{1,1}$. During the process, $\tau$ preserves adjacency and covers all possible translated positions where 
$F$ can act. Given such a translation $\tau$, we can define the product $$\prod_{\tau}F:=\prod_{i=0}^{|L|-1}F_{\tau^i(U)}=\prod_{i=0}^{|L|-1}F\circ\tau^i.$$ Hence $\tau$ provides an \textit{order} of composing local operators. Since we assume the commutativity of those operators we are working with, we shall not worry about the order of the composition.
\\
This construction obviously generalizes to an arbitrary dimension and arbitrary length of the sides of $U$. 
\\
\newline
\noindent Given any tensor $F^{i_1,\dots,i_k}_{j_1,\dots,j_l}$, there is a natural action of $S_{k+l}$ on $F$. We denote the action by $\sigma\bullet F$ for $\sigma\in S_{k+l}$. Note that if $\sigma$ is a \textit{$k,l$-unshuffled} permutation, the type of $F$ is preserved. We call such actions \textbf{permutations of the indices of $F$}. 
\\
Now fix $d=1$, so $|L|=N$. Given any local Hamiltonian $\mathcal{H}_{loc}$ acting locally on two adjacent sites $U$ (so $U\subseteq L,|U|=2$) with associated projection $\mathrm{P}_U$, the product 
$$\prod_{x\in L}\mathrm{P}_{U+x}$$
is well defined by our assumption on their commutativity.
\\
\newline

\begin{lem} With above setting, 
    $$\tr \prod_{x\in L} \mathrm{P}_{U+x}=\tr ((\sigma\bullet \mathrm{P}_U)^{N})$$
for some permutation $\sigma$ of the indices of $\mathrm{P}$.
\end{lem}
\begin{proof}
    We denote $\mathrm{P}_{U}:V\otimes V\rightarrow V\otimes V$ mapping $e_{i,j}\mapsto \mathrm{P}_{ij}^{i'j'}e_{i',j'}$, where $e_{\alpha,\beta}=e_{\alpha}\otimes e_{\beta}$.\\
    Then by translating $U$,
    \begin{align}
        LHS=&Tr(\mathrm{P}_{1,2}\mathrm{P}_{2,3}\dots \mathrm{P}_{N-1,N}\mathrm{P}_{N,1})\\
        =&\sum \mathrm{P}^{i_1i_2'}_{i_1'i_2}\mathrm{P}^{i_2i_3'}_{i_2'i_3}\dots \mathrm{P}^{i_Ni_1'}_{i_N'i_1}\label{traceform}
    \end{align}
    Let 

    \begin{align*}
        \alpha&=
        \begin{pmatrix}
            i\\
            i'
        \end{pmatrix}
    \end{align*}
    and 
    \begin{align*}
    \beta&=
    \begin{pmatrix}
        j\\
        j'
    \end{pmatrix}
    \end{align*}
    and define the tensor $Q$ as $Q^{\alpha}_{\beta}:=\mathrm{P}^{ij'}_{i'j}$,
    Observe that $Q=\sigma\bullet \mathrm{P}$, where $\sigma$ is the swapping of the two indices $i$ and $i'$. On the other hand we have
    \begin{align*}
        (\ref{traceform})&=\sum Q^{\alpha_1}_{\alpha_2}Q^{\alpha_2}_{\alpha_3}\dots Q^{\alpha_N}_{\alpha_1}\\
           &= Tr(Q^N).
    \end{align*}
    
\end{proof}
\noindent\textit{Remark.} This argument obviously generalizes to larger $\mathcal{H}_{loc}$ acting locally on adjacent $m$ sites for arbitrary $m\leq N$ when $d=1$; also to higher dimension: we treat each dimension separately for the translation $\sigma$, adapting the proof with $|L|=N^d$ and making a change of the tensor indices accordingly. Note that $LHS$ of the equation in the lemma is independent of the permutation $\sigma$, whereas the permutation $\sigma$ is itself independent of the size of the lattice $N$.\\
\newline
Given local Hamiltonian $\mathcal{H}_{U}$ acting on $U\subseteq L$, we have
$$\mathrm{P}_U=\lim_{\beta\rightarrow\infty}e^{-\beta \mathcal{H}_U},$$
On the other hand, the global Hamiltonian 
$$\mathcal{H}=\sum_{U'\subseteq L \mathrm{ \,is\, a\, translation\, of\,} U}\mathcal{H}_{U'}$$ whose projection to the kernel is
\begin{align*}
    \mathrm{P}&=\lim_{\beta\rightarrow\infty}e^{-\beta(\sum_{U'} \mathcal{H}_{U'})}\\
     &=\lim_{\beta\rightarrow\infty}\prod_{U'} e^{-\beta \mathcal{H}_{U'}}\\
     &=\prod_{U'}\mathrm{P}_{U'}
\end{align*}
where $U'$s' in the equation are translations of $U$ in $L$. 
%要改的
Because of the assumption on $[\mathcal{H}_{U'},\mathcal{H}_{U''}]$, together with $\mathcal{H}$ being a sum of local Hamiltonians is independent of the order of the corresponding local sites, so does the projection $\mathrm{P}$.  
%要改的
\\
\newline
\textit{Proof of Theorem for $d=1$.} Given lattice $L$ of size $N$, local Hamiltonian $\mathcal{H}_{loc}$ acting on $U\subseteq L$, and a permutation $\sigma$ on the indices of $\mathrm{P}_U$, we have 
\begin{align*}
    G_N&=\tr(\prod_{U'}\mathrm{P}_{U'})=\tr(\prod_{x\in L}\mathrm{P}_{U+x})\\
       &=\tr((\sigma\bullet \mathrm{P}_{loc})^N)
\end{align*}
Set $A:=\sigma\bullet \mathrm{P}_{loc}$, then $G_N=\tr(A^N)$. The Caylay-Hamilton theorem \cite{axler1997linear} shows that $$\chi_A(A)=0,$$ where $\chi_A$ is the (monic) characteristic polynomial of $A$. Moreover, since $\mathrm{P}_{loc}$ is a projection and $\sigma\bullet\square$ is a permutation of the tensor coordinates, $A$ has integer entries. It follows that $\chi_A$ is an integral polynomial\footnote{A polynomial is integral if it is monic and has integer coefficients}. Let $$\chi_A(x)=a_{N_0}+a_{N_0-1}x+\dots +a_1x^{N_0-1}+x^{N_0}$$ for $a_i\in\mathbb{Z}$. Taking the trace of $\chi_A(A)$ proves the result we want. $\square$\\
%#########################靠，怎么编下去#############################
\newline
\noindent
\textit{Remark.} This proof can be adapted to any higher dimensional lattice by treating each dimension separately.

\begin{cor}
     In the low temperature limit $T\longrightarrow 0$, the von Neumann entropy $S_N$ satisfies a recurrence relation of order $N_0$ depending only on $\dim V$ and $m$ (the size of the sublattice $U$), not on $N$.
\end{cor}
\begin{proof}
    Directly from Theorem 1.
\end{proof}
\begin{cor}
    In the low temperature and thermodynamical limit  $T\longrightarrow 0$, $N\longrightarrow \infty$, the \textbf{average entropy per site}

$$\lim_{N\rightarrow\infty,T\rightarrow 0} \frac{S_N}{N} $$
is the logarithm of an algebraic integer, \ie its exponential is a root of a monic polynomial equation $$\sum_{k=0}^{N_0} a_k x^{N_0-k}=0$$ with integral coefficients.

\end{cor}
\begin{proof}
    From the proof of Theorem 1, $G_N$ satisfies a recurrence relation 
        $$\sum_{k=0}^{N_0} a_kG_{N-k}=0$$
    where all coefficient $a_k\in\mathbb{Z}$, and $a_{k_0}=\pm1$ for the smallest $k_0\in\{0,1,\dots,N_0\}$ such that $a_{k_0}\neq 0$. Assume $k_0=0$ without loss of generality. The charactristic root technique \cite{grimmett2014probability} says that, if $\theta_1,\dots,\theta_r$ are distinct roots of the equation  
        $$\sum_{k\leq N_0} a_kx^{N_0-k}=0$$
    with multiplicity $m_1,\dots,m_r$, respectively, then 
    \begin{align*}
        G_N&=(c^1_1+c^1_2N+\dots+c^1_{m_1}N^{m_1-1})\theta_1^N\\
        &+(c^2_1+c^2_2N+\dots+c^2_{m_2}N^{m_2-1})\theta_2^N\\
        &+\dots\\
        &+(c^r_1+c^r_2N+\dots+c^r_{m_r}N^{m_r-1})\theta_r^N
    \end{align*}
    for some constants $c^i_j$'s. 
    Notice that $\theta_1,\dots,\theta_r$ are all algebraic integers being roots of the monic equation $\sum_{k\leq N_0} a_kx^{N_0-k}=0$. \\
    \newline
    There are two basic facts:
    \begin{itemize}
        \item For any constants $m\in\mathbb{N}^*$ and $c_1,\dots,c_m$ (where not all $c_j$'s are zero),
        $$\lim_{N\rightarrow\infty}\frac{1}{N}\log(c_1+c_2N+\dots+c_{m}N^{m-1})=0;$$
        \item If $|x|>|y|\geq 0$, then 
        $$\lim_{N\rightarrow\infty}\frac{x^N+y^N}{x^N}=1 $$
    \end{itemize}
    we deduce that 
    $$\exp(\lim_{N\rightarrow\infty}\frac{S_N}{N})=\lim_{N\rightarrow\infty}\sqrt[N]{G_N}=\theta_l$$
    where $\theta_l$ has the maximum modulus among $\theta_1,\dots,\theta_r$.    
\end{proof}

\noindent
We have essentially shown that:
\begin{thmA}\label{thmA}
    Given a sequence of finite-dimensional vector spaces each isomorphic to $V$ and a projection $\mathrm{P}\in \End(V^{\otimes 2})$ such that $$[\mathrm{P}_{i,i+1},\mathrm{P}_{j,j+1}]=0\in \End(V\otimes V\otimes V),$$ then 
    \[
\mathrm{P}(N)=\mathrm{P}_{1,2}\mathrm{P}_{2,3}\dots \mathrm{P}_{N-1,N}
    \] 
    is a projection and $$\exp(\lim_{n\rightarrow\infty}\frac{1}{N}\log \dim \Image\mathrm{P}(N))$$ is an algebraic integer.
\end{thmA}

\section{Categorical entropy of $A_{\infty}$-categories}
Given an exact endofunctor \( F \) on a triangulated category with a generator \( G \), the categorical entropy \( h_t(F) \) can be expressed as:
\[
h_t(F) = \lim_{n \to \infty} \frac{1}{n} \log \delta_t(G, F^n G),
\]
where \( \delta_t \) measures the growth in complexity of objects transformed under \( F \), using relations to the generator \( G \) \cite{1307.8418}. In the \( A_\infty \)-category context, the entropy of an endofunctor can often be computed in terms of the Poincaré polynomial of Ext-groups, particularly when dealing with saturated categories \cite{1307.8418}. For example, in the case of smooth projective varieties or Fukaya categories \cite{HoriMS}, categorical entropy offers insights into the dynamical properties of endofunctors, relating directly to the topological entropy of associated maps in classical dynamical systems \cite{Fathi,kim2022computationcategoricalentropyspherical,FriedlandEntropy}.
\\
\newline
\noindent
In the framework of saturated \( A_\infty \)-categories \cite{0606241v3}, categorical entropy often captures the exponential growth rate of various features of the category, such as the dimensions of Ext-groups or the spectral radius of actions on Hochschild homology. This connection enables categorical entropy to reveal the underlying complexity, stability, and dynamical behavior of objects within the category as they evolve under iterated functorial actions. When the entropy is constant, the structure of the category under \( F \) may be relatively stable. In contrast, larger values or varying entropy often point to more chaotic or intricate transformations, where objects in the category grow in complexity with repeated applications of the functor \cite{1307.8418}.
\newline
\\
\noindent
\textbf{Convention.} We adopt the \textit{Koszul sign rule}: the choice of signs will be dictated by the principle that whenever we switch two objects of degrees $p$ and $q$, respectively, we multiply the sign by $(-1)^{pq}$. More precisely, given graded maps of graded vector spaces $f,g:V\rightarrow W$, if $v_1,v_2\in V$ are homogeneous elements, $f\otimes g(v_1\otimes v_2)=(-1)^{|g||v_1|}f(v_1)\otimes g(v_2)$. Under this setting, the commutator $[f,g]:=f\circ g -(-1)^{|f||g|}g\circ f$.

\subsection{Categorical entropy of saturated $A_{\infty}$-cateegories}
Here we briefly recall the construction of categorical entropy. Then, we discuss the algebraicity of the categorical entropy when $t=0$ and try to give a sufficient condition. For details of $A_{\infty}$-categories and triangulated categories, see \cite{0510508,1307.8418,0606241v3}. We assume that the $A_{\infty}$-category is over a fixed field $k$.% (and we assume $k\cong\mathbb{C}$ unless stated otherwise). 
\\

\noindent
Fix a triangulated category $\mathcal{D}$ and an object $G$ in $\mathcal{D}$. Given any object $E$ in $\mathcal{D}$, consider towers of triangles of the following form:
\begin{equation}
\begin{tikzcd}[column sep=1ex,row sep=1ex]
  0 \arrow[rr] && A_1 \arrow[dl] \arrow[rr]&& A_2 \arrow[dl]\cdots A_{k-1}\arrow[rr] && A_k\arrow[dl]    \cong E\bigoplus E'            \\
  &G[n_1]  \arrow[ul,dashed]&& G[n_2] \arrow[ul,dashed]&& G[n_k] \arrow[ul,dashed]
\end{tikzcd}\label{complexity}
\end{equation}
for some $E'$. The \textbf{complexity} of $E$ relative to $G$ is given by 
\[
\delta_t(G,E)=\inf\{\Sigma^k_{i=1}e^{n_it}|\exists E'\textit{ so that some tower of triangles of form (\ref{complexity}) holds}.\}
\]
Of course, if such tower doesn't exist, $\delta_t(G,E)$ is set to be $\infty$. If such tower exists for any object $E$, $G$ is said to be a (split-)generator of $\mathcal{D}$. For given $G$ and $E$, we can also regard the complexity as a function $\delta_{\square}(G,E):\mathbb{R}\rightarrow[-\infty,+\infty]$.\newline
\\
For objects $E_1$, $E_2$, $E_3$ in $\mathcal{D}$, the complexity functions satisfy the following \cite{1307.8418}:
\begin{itemize}
    \item \textit{(triangle inequality):} $\delta_t(E_1,E_3)\leq\delta_t(E_1,E_2)\delta_t(E_2,E_3)$;
    \item \textit{(subadditivity):} $\delta_t(E_1,E_2\bigoplus E_3)\leq\delta_t(E_1,E_2)+\delta_t(E_1,E_3)$;
    \item \textit{(retraction):} $\delta_t(F(E_1),F(E_2))\leq\delta_t(E_1,E_2)$ for any exact functor of triangulated categories $\
    \mathcal{D}\rightarrow\mathcal{D}'$.
\end{itemize}
If $G$ is a generator of $\mathcal{D}$ and $F$ is an exact endofunctor of $\mathcal{D}$, 
%it is expected \cite{1307.8418} that $\delta_t(G,F^N(G))$ grows at most exponentially with $N$. 
the entropy of $F$ is defined by the exponent of $\delta_t(G,F^N(G))$:
\[
h_{t}(F,G):=\lim_{N\rightarrow\infty}\frac{1}{N}\log(\delta_{t}(G,F^N(G)))
\]

\noindent
It is shown \cite{1307.8418} that this limit is independent of the choice of generator $G$.
\newline
\\
\noindent
We mainly focus on \textbf{saturated} $A_{\infty}$-categories. An $A_{\infty}$-category $\mathcal{C}$ is said to be saturated if it is triangulated and is Morita equivalent to a smooth and compact $A_{\infty}$-algebra \cite{0606241v3}. Under this setting, the entropy of an endofunctor $F\in \End(\mathcal{C})$ is computable. Moreover, if $G$ is a generator of $\mathcal{C}$, then \cite{1307.8418}
\[
h_t(F)=\lim_{N\rightarrow\infty}\frac{1}{N}\log\sum_{n\in\mathbb{Z}}\dim \Ext^n(G,F^NG)e^{-nt}.
\]
Also (see \cite{0510508,0310337} for detailed proofs):
\begin{itemize}
    \item Denote $R:=\End_{\mathcal{C}}(G)$, then $R\lmod\cong \mathcal{C}$ via 
    \begin{align*}
        &M \mapsto M\otimes_R G\in\mathcal{C} \\
        &\Hom_{\mathcal{C}}(G,X)\mapsfrom X
    \end{align*}
    \item A functor $F:\mathcal{C}\rightarrow\mathcal{C}$ corresponds to $F_M:M\otimes_{R}\square:R\lmod\rightarrow R\lmod$ for some $R$-bimodule $M$. 
\end{itemize}
It follows that
\begin{align}
    h_0(F) &= \lim_{N\rightarrow\infty}\frac{1}{N}\log\sum_{n\in\mathbb{Z}}\dim \Ext^n(G,M^{\otimes_R N}\otimes_R G)\label{eq1}
\end{align}
It is conjectured in \cite{1307.8418} that $\exp(h_0(F))$ is an algebraic integer. 
\\
\newline
\noindent
Take a free resolution $(R\otimes_kC\otimes_kR,d)$ of $M$ as an $R$-bimodule with some vector space $C$, then $C$ is naturally equipped with grading and 
\begin{align}
(\ref{eq1}) &=\lim_{N\rightarrow\infty}\frac{1}{N}\log\sum_{n\in\mathbb{Z}}\dim H^n(\Hom_R(R,(R\otimes_kC\otimes_kR)^{\otimes_RN}))\\
    &= \lim_{N\rightarrow\infty}\frac{1}{N}\log\sum_{n\in\mathbb{Z}}\dim H^n((R\otimes_kC\otimes_kR)^{\otimes_RN},d^{\otimes N})\\
    &= \lim_{N\rightarrow\infty}\frac{1}{N}\log\sum_{n\in\mathbb{Z}}\dim H^n((R\otimes_kC)^{\otimes_kN}\otimes_kR,d^{\otimes N})
\end{align}
because $R\otimes_RR\cong R$. Here 
$$d^{\otimes N}:=\sum^N_{\substack{i=1 \\ \textit{$d$ is at the}\\\textit{ $i$-th place}}} \Id\otimes\dots\otimes \Id\otimes d\otimes \Id\otimes\dots\otimes \Id$$
is the differential on $(R\otimes_kC)^{\otimes_kN}\otimes_kR\cong (R\otimes_kC\otimes_kR)^{\otimes_RN}$.
\\
\newline
\noindent
Now fix $V:= R\otimes_kC$. Consider a one-dimensional lattice of size $N$ whose local sites are associated with $V_1\cong V_2\cong\cdots\cong V_N\cong V$, respectively (and set $V_{N+1}:=V_1$). Then the action of the differential $d$ on $R\otimes_kC\otimes_kR$ induces a local differential operator $d_{loc}:=d\otimes \Id_C$ acting on $V\otimes_kV\cong R\otimes_kC\otimes_kR\otimes_kC$. Assume that $d_{loc}$ acts on $V_1\otimes_kV_2$ and that $\tau$ is a translation which preserves adjacency with respect to $d_{loc}$, it follows that the local differentials commute: $[d_{loc},d_{loc}\circ\tau]=0$.\newline
\\
\noindent
On the other hand, let $\partial:=\sum_{i=0}^{N-1}d_{loc}\circ\tau^i$ be the sum of all such local operators. Then $\partial$ is the differential operator of the total space (which is a chain complex) $V^{\otimes_kN}\otimes_kV$, where The extra $V$ comes from $V_{N+1}=V_1$. Therefore, by the commutativity of local differentials,
\begin{align}
 \sum_{n\in\mathbb{Z}}\dim H^n((R\otimes_kC)^{\otimes_kN}\otimes_kR,d^{\otimes N})= \sum_{n\in\mathbb{Z}}\dim H^n(V^{\otimes_kN}\otimes_kV,\partial).\label{dimH}
\end{align}

\begin{figure}[H]
    \centering
    \includegraphics[width=10cm]{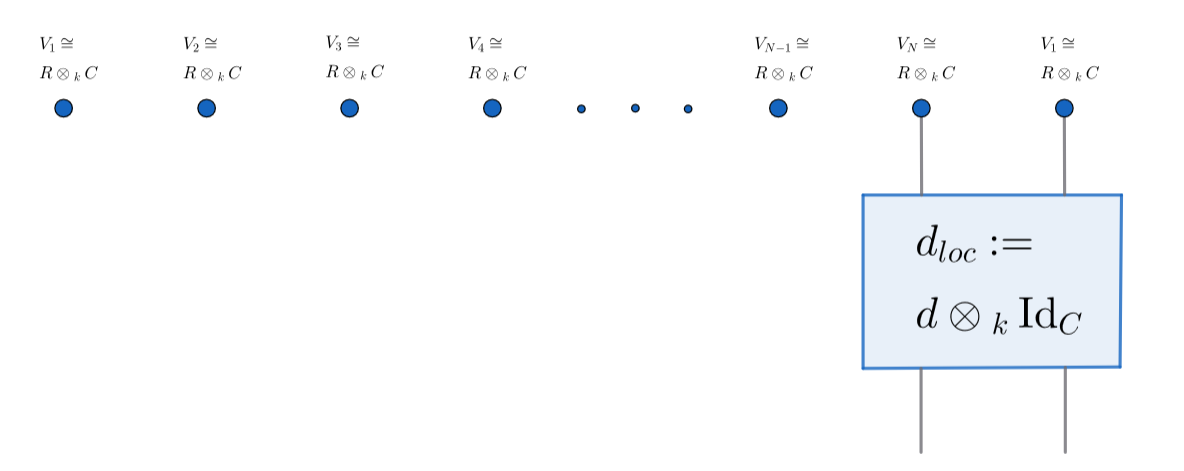}
    \caption{The action of $d_{loc}$ on the adjacent sites $V_N$ and $V_1$. This contributes the extra $R$ in the $LHS$ as well as the extra $V$ in the $RHS$ of (\ref{dimH}).}
    \label{fig:graphic representation example}
\end{figure}
\noindent
We hereby propose a statement:
\begin{condB}\label{Claim B}
    Given a sequence of finite dimensional graded vector spaces each isomorphic to $V$, and a differential $Q\in \End(V\otimes V)$ homogeneous of degree $1$ such that $$[Q_{i,i+1},Q_{j,j+1}]=0\in \End(V\otimes V\otimes V),$$ %(here $[-,-]$ means super commutator in the sense of super Lie algebra), 
    then 
    \[
    Q(N):=Q_{1,2}+Q_{2,3}+\dots+Q_{N-1,N}\in \End(V^{\otimes N})
    \]
    is a differential and 
    \[\exp(\lim_{N\rightarrow\infty}\frac{1}{N}\log\sum_{n\in\mathbb{Z}}  \dim H^n(Q(N)))\] 
    is an algebraic integer.    
\end{condB}
\noindent
Clearly, this statement implies the algebraicity of $\exp(h_0(F))$.
\\
\newline
\noindent
\textit{Remark.} The above condition can be adapted to super(\ie $\mathbb{Z}_2$-graded)-vector spaces. Under such conditions, the average entropy per site becomes \[\frac{1}{N}\log( \dim H^0(Q(N))+\dim H^1(Q(N))). \]

\subsection{Categorical entropy from lattice model}
To relate the two concepts of entropy from different areas, we need first to generalize the concept of lattice model to consider gauge symmetry. The basic ingredients of a quantum lattice model with gauge symmetry are the following:
\begin{itemize}
    \item A $d$-dimensional lattice $L$ of size $N$ (Here by lattice we mean $\{1,2,\cdots,N\}^d$)

    \item For each lattice point $x\in L$, we have a local graded Hilbert space $V_x$ such that all of them are canonically isomorphic (here we regard, for example, $V_1=V_{N+1}$ when $d=1$, and similarly for higher dimensions).

    Given a subset $U\subseteq L$, we define its space of states 

    $$V(U)=\bigotimes _{x\in U} V_x$$

    and the global space of states $V=V(L)$
    
    \item For some subset $U\subseteq L$ of size $m<N$, we have a local Hamiltonian (\ie a non-negative definite Hermitian operator) $\mathcal{H}_U$ and a local BRST transformation (\ie differential \cite{fuster2005brstantifieldquantizationshortreview}) $Q_U$ acting on $V(U)$ such that the commutator

    $$[Q_U,Q_{U'}]=[\mathcal{H}_U,Q_{U'}]=0$$ for $U,U'\subseteq L$.

    Here the first commutator comes from the previous section, and the second commutator says precisely that the gauge transformation respects the dynamical evolution of the system:\\
\begin{figure}[H]
    \centering
    \includegraphics[width=7.7cm]{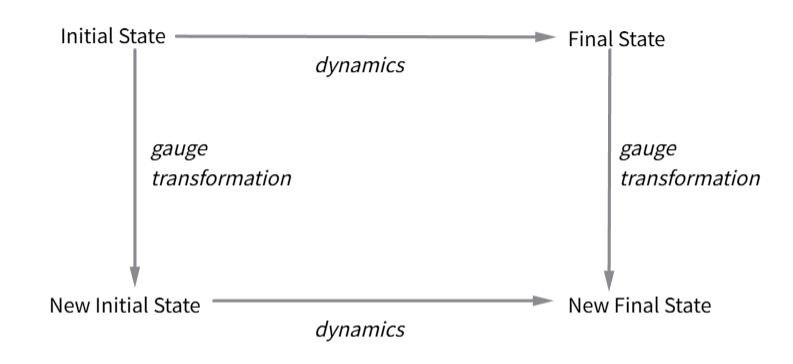}
    \caption{The gauge transformation respects the dynamics}
    \label{Gauge symmetry}
\end{figure}
    Since $\mathcal{H}_U$ and $Q_U$ act on $V(U)$, they also naturally act on $V$ by tensoring with the identity operators on other sites. 
    
    We define the global Hamiltonian 

    $$\mathcal{H}=\sum _{U'\subseteq L \mathrm{ \,is\, a\, translation\, of\,} U} \mathcal{H}_{U'}$$

    and the global gauge transformation 

    $$Q=\sum _{U'\subseteq L \mathrm{ \,is\, a\, translation\, of\,} U} Q_{U'}$$

\end{itemize}
\noindent
We may directly check that $Q$ is a differential on $V$ and $\mathcal{H}$ descends to a non-negative definite Hermitian operator $\mathcal{H}^{phys}$ on the space of physical states defined as the $Q$-cohomology on $V$: $$V^{phys}=H^*(V,Q).$$
\noindent
For $\mathcal{H}^{phys}$ we may also define the von Neumann entropy as in previous sections. There are two special situations where the local operators acts on two adjacent sites (so $U\subseteq L$ has size $m=|U|=2$) are of interest:
\begin{enumerate}
    \item $Q_U=0$: This implies the case for a lattice model, where 
  $\mathcal{H}^{phys}=\mathcal{H}$ and the von Neumann entropy is $S=\log(\dim \ker \mathcal{H})$. In the limit $N\rightarrow\infty$, the average entropy per site is 
    $$\lim_{N\rightarrow\infty}\frac{1}{N}\log\dim \Image\mathrm{P}(N).$$
    We have shown that this is the logarithm of some algebraic integer;
    \item $\mathcal{H}^{phys}=0$: In this case, $V^{phys}=\bigoplus_{n\in\mathbb{Z}}H^n(V,Q)$. Since the global (induced) Hamiltonian vanishes, $\mathrm{P}^{phys}=\Id$, hence
    \begin{align*}
        \dim \Image\mathrm{P}^{phys}(N) &= \dim \Image (\Id_{V^{phys}})\\
                            &=\dim V^{phys}\\
                            &= \sum_{n\in\mathbb{Z}}\dim H^n(V,Q)       
    \end{align*}
    So that the average entropy per site in the limit $N\rightarrow\infty$ is 
    \[
    \lim_{N\rightarrow\infty}\frac{1}{N}\log\sum_{n\in\mathbb{Z}}  \dim H^n(Q(N));
    \]
    If this is the logarithm of some algebraic integer, it would imply the algebraicity of $\exp(h_0(F))$, as stated in \ref{Claim B}.
\end{enumerate}
We propose the analogue of our main result about the von Neumann entropy in the setting of gauged lattice models, whose proof, however, remains unknown:
\begin{mconj}
In the low-temperature and thermodynamic limits $T\longrightarrow 0$, $N\longrightarrow \infty$, the average von Neumann entropy per site
$$\lim_{N\rightarrow\infty,T\rightarrow 0} \frac{S^{phys}_N}{N}=
\lim_{N\rightarrow\infty}\frac{1}{N}\log\dim(\ker \mathcal{H}_N^{phys})$$
of a \textbf{gauged lattice model} of above setting is the logarithm of an algebraic integer.
\end{mconj}
\noindent
From the above discussion, it is immediate that
\begin{thm}
    Categorical entropy corresponds to the von Neumann entropy of a gauged lattice model. Hence our main conjecture for the gauged lattice model would imply the algebraicity of categorical entropy in the case $t=0$ of a saturated $A_{\infty}$-category conjectured in \cite{1307.8418}.
\end{thm}

\noindent
\textit{Remark.} The above setup can be adapted to super-vector spaces. Under such condition, one changes the formula for the average entropy per site accordingly, and the commutators are replaced by super-commutators. 

\section{Conclusion}

In this paper, we have explored the connections between the concept of entropy in statistical mechanics and the more recent notion of categorical entropy, particularly in the context of \( A_\infty \)-categories. By examining the parallels between the entropy-driven behavior of physical systems and the dynamical properties of endofunctors in triangulated and \( A_\infty \)-categories, we have shown that both concepts reflect fundamental aspects of complexity, randomness, and stability in their respective domains. 
\\
\newline
\noindent
The Boltzmann ensemble in classical and quantum statistical mechanics exemplifies how entropy maximization provides insights into the equilibrium states of physical systems. Similarly, categorical entropy offers a measure of the growth in complexity associated with repeated transformations in a category, drawing a direct analogy to entropy in physical systems. In particular, we have investigated how categorical entropy in \( A_\infty \)-categories can be computed using the spectral radius of actions on Hochschild homology, highlighting a pathway for understanding the homological and dynamical aspects of these categories.
\\
\newline
\noindent
Furthermore, we extended our discussion to lattice models, both in classical and quantum settings, as simplified representations of interacting particle systems. We demonstrated that, in the low-temperature limit, the von Neumann entropy of these models exhibits distinct characteristics that can be studied systematically. This analysis has also provided insights into the behavior of entropy in the thermodynamic limit, paving the way for further exploration of lattice models with gauge symmetry. By connecting categorical entropy with gauged lattice models, we propose a conjectural link between the algebraic properties of categorical entropy and the von Neumann entropy within these models.
\\
\newline
\noindent
This work provides a foundation for future interdisciplinary research, suggesting that categorical entropy may play a role in understanding entropy beyond physical systems. In particular, the parallels we draw between categorical and physical entropy invite further exploration into the applications of categorical entropy in other areas, such as symplectic geometry and algebraic geometry, where the \( A_\infty \)-framework is prevalent. Additionally, our findings may inform studies in quantum information theory, where von Neumann entropy and its categorical analogues can offer new perspectives on information, entanglement, and complexity. The connections outlined in this paper highlight the potential for a unified framework where entropy serves as a central concept across both mathematical and physical theories.

\printbibliography

\end{document}